\documentclass[paper,letterpaper]{ieice}
\usepackage{amsmath,amssymb}
\usepackage[T1]{fontenc}
\usepackage{graphicx}
\usepackage{amsthm}
\usepackage[varg]{txfonts}
\usepackage{color}
\allowdisplaybreaks
\field{A}
\vol{93}
\no{11}
\SpecialSection{Information Theory and Its Applications}
\title{Vulnerability of MRD-Code-based Universal\\Secure Network Coding against Stronger Eavesdroppers}
\authorlist{
 \authorentry[shioji@it.ss.titech.ac.jp]
 {Eitaro SHIOJI}{n}{labelA}[present affiliate label]
 \authorentry[ryutaroh@rmatsumoto.org]
 {Ryutaroh MATSUMOTO}{r}{labelA}
 \authorentry[uyematsu@ieee.org]
 {Tomohiko UYEMATSU}{e}{labelA}
}
\affiliate[labelA]{The authors are with Dept.~of Communications and Integrated Systems, Tokyo Institute of Technology, Tokyo, 152-8552 Japan.}
\paffiliate[present affiliate label]{Presently, the author is with NTT Information Sharing Platform Laboratories, NTT Corporation, Musashino-shi, 180-8585 Japan.}
\titlenote{A part of this paper was presented at 2010 IEEE International Symposium on Information Theory.}

\received{2010}{2}{19}
\revised{2010}{5}{20}

\newtheorem{definition}{Definition}
\newtheorem{theorem}{Theorem}
\newtheorem{lemma}{Lemma}

\newtheorem{corollary}{Corollary}

\begin{document}
\maketitle
\begin{summary}
Silva et al.\ proposed a universal secure network coding scheme based on MRD codes, which can be applied to any underlying network code. This paper considers a stronger eavesdropping model where the eavesdroppers possess the ability to re-select the tapping links during the transmission. We give a proof for the impossibility of attaining universal security against such adversaries using Silva et al.'s code for all choices of code parameters, even with a restricted number of tapped links. We also consider the cases with restricted tapping duration and derive some conditions for this code to be secure.
\end{summary}

\section{Introduction}

The notion of network coding, proposed by Ahlswede et al.~\cite{AHLSWEDE}, has been attracting much attention. On a conventional routing network, each node is only allowed to relay the received packets to the next node, while on a network with network coding support, each node is allowed to perform some data processing using the received packets and send the result to the next node. It is known that the use of network coding offers many advantages over the use of conventional network, such as achievement of higher rate in multicast communications or better energy efficiency in wireless communications~\cite{NCF}.

Secrecy of communication, or more specifically, information-theoretically secure communication in the presence of an adversary capable of tapping a fixed number of links of its choice, is considered as one of such advantages of network coding. Such a scheme, referred to as secure network coding, consists of the following two components: the network code which determines how packets are coded at intermediate nodes, and the outer code which is a pre-coding done at the source before transmission. Several secure network codes have been proposed, such as the one by Cai et al.~\cite{CAI}. However, these codes require the reconstruction of the network code, or the reconstruction of the outer code in order to attain security for a given set of tapped links. Such a property causes problems, such as difficulty when securing random network codes~\cite{RNC}, where network codes are constructed randomly.

Silva et al.\ proposed a universal secure network coding method~\cite{SILVA} based on MRD codes~\cite{GABIDULIN} and coset coding scheme~\cite{WYNER}. This code can be applied on top of any already-constructed network code to attain security. However, due to its use of vector outer code which requires that each symbol be transmitted over multiple time slots, it must assume that the tapped links are fixed during the transmission period. 

We consider a stronger eavesdropping model where the eavesdroppers possess the ability to re-select the tapping links during the transmission. Such a model is worth consideration because the conventional non-universal secure network codes (e.g.~\cite{CAI},~\cite{ROUAYHEB}) are guaranteed to be secure against it. Moreover, this model corresponds to some practical situations where random network coding is used and the coding vectors are time-varying, such as the robust random network coding scheme proposed by Chou et al.~\cite{CHOU}. Also, the current standard of the IP protocol allows the network to split a packet into multiple fragments and carry them through multiple distinct routes, as explained in \cite[Section 11.5]{RW94}. Thus, the stronger eavesdropping model considered here has practical importance when Silva et al.'s method is used over the current Internet.

This paper aims to clarify the security of Silva et al.'s universal code against this eavesdropping model, and is organized as follows. In Section 2 we define some notations and briefly review some of the existing results of secure network coding, and describe Silva et al.'s universal secure network code. In Section 3 we introduce our stronger eavesdropping model. In Section 4 we prove the vulnerability of this universal code against our model for all code parameters. We also prove that the code is vulnerable even with a limited number of tapped links. Moreover, the cases with shorter tapping duration are considered, and sufficient conditions and necessary conditions for the code to be secure are given. In Section 5 we state our conclusion and the future tasks. 

\section{Preliminaries}

In this section we define our basic notations and review some of the existing results of secure network coding.

\subsection{Extension Field}

The extension field $\mathbb{F}_{q^m}$ of $\mathbb{F}_q$ can be regarded as a vector space over $\mathbb{F}_q$. Thus, when the basis of this space is fixed, an element of $\mathbb{F}_{q^m}$ can be represented as an $m$-dimensional vector over $\mathbb{F}_q$. For $y \in \mathbb{F}_{q^m}$, denote the $i$-th element of its vector representation as $y^{(i)}$. Accordingly, the vector representation of $x \in \mathbb{F}_{q^m}$ is written as $(x^{(1)}, x^{(2)}, \cdots , x^{(m)}) \in \mathbb{F}_q^m$.

\subsection{Network Coding}

Data communication over a network is considered. We use a network model defined by an acyclic and directed graph $G = (V,E)$, where $V$ and $E$ denote the set of nodes and the set of links, respectively. In this model we assume that, each link can carry an element of $\mathbb{F}_q$ per unit time, and data flowing on the network is not affected by delays, erasures or errors. 

Let $s \in V$ and $\mathcal{R} \subset V$ denote the source node and the set of sink nodes, respectively. The source node wishes to multicast the sequence $X = (X_1, X_2, \cdots, X_n)^T \in \mathbb{F}_q^n$ to all sink nodes at rate $n$. The rate is defined as the number of elements of $\mathbb{F}_q$ transmitted at the source node per unit time. Assume $n \leq \text{min}\{\text{maxflow}(s,r)\mid r \in \mathcal{R}\}$ holds, where $\text{maxflow}(i,j)$ denotes the maximum flow from node $i$ and $j$. We assume that linear network coding \cite{LNC} is employed on the network, i.e. the type of data processing performed on the packets at each node is limited to linear combination. This implies that the data flowing on any link on the network can be represented as an $\mathbb{F}_q$-linear combination of the sequence $X_1, X_2, \cdots,X_n$. Thus, the information flowing on a link $e$ can be denoted as $Y_e= \vec{b}_e \cdot X$ using a global coding vector (GCV), $\vec{b}_e = (b_1, b_2, \cdots , b_n)^T \in \mathbb{F}_q^n$, where ``$\cdot$" denotes the inner product operator for vectors. When one has access to, say, the $l$ links $e_1, e_2, \cdots, e_l$, then the information obtained from these links is denoted as $MX \in \mathbb{F}_q^l$, where $M = (\vec{b}_{e_1},\vec{b}_{e_2},\cdots, \vec{b}_{e_l})^T$.

Constructing a network code is equivalent to fixing the GCV of each link by setting the coefficients of the linear combination performed at each node. A network code is called feasible if every sink is able to decode $X$. When $q$ is sufficiently large, a feasible network code for rate $n$ multicast can always be constructed~\cite{NCF}.

\subsection{Secure Network Coding} \label{sec:SNC}

The wiretap network model used in the works \cite{CAI} and \cite{SILVA} on which secure network coding is employed is described below. For simplification, only one receiver is assumed. Let $F$ be some extension field of $\mathbb{F}_q$.

\begin{itemize}
\item \textbf{Sender:} The sender wishes to send the secret information sequence represented by a random variable $S = (S_1, S_2, \cdots, S_k)^T$ distributed uniformly over $F^k$. $S$ is first coded into the sequence $X = (X_1, X_2, \cdots, X_n)^T \in F^n$ using an outer code and then $X$ is sent over the network with a feasible network code.
\item \textbf{Receiver:} The receiver receives the information sequence $Y = AX = (Y_1, Y_2, \cdots, Y_n)^T \in F^n$, where $A$ is the matrix constructed by appending the GCVs of the input links to the receiver node.
\item \textbf{Eavesdropper:} The eavesdropper is able to wiretap any $\mu$ links on the network. Let the set of tapped links be $\mathcal{I} =  \{e_1, e_2, \cdots, e_\mu\} \subseteq E$. Then the wiretapped information sequence is represented as $W = BX = (W_1, W_2, \cdots, W_\mu)^T \in F^\mu$ using the matrix $B = (\vec{b}_{e_1}, \vec{b}_{e_2}, \cdots, \vec{b}_{e_\mu})^T \in \mathbb{F}_q^{\mu \times n}$.
\end{itemize}

The security which guarantees that no information about $S$ leaks out to the wiretapper even when $\mu$ arbitrary links are wiretapped, is defined as follows.

\begin{definition}[strong security~\cite{CAI}] \label{def:strong}
\begin{eqnarray}
&&H(S|Y) = 0\label{eqn:decodability},\\
&&I(S;W=BX) = 0, \forall \mathcal{I} \subseteq E,|\mathcal{I}| = \mu. \nonumber
\end{eqnarray}
\end{definition}

\noindent Condition (\ref{eqn:decodability}) is satisfied if the outer code used is uniquely decodable and the network code used is feasible. Cai et al.\ showed a construction method~\cite{CAI} for secure network codes that satisfies the conditions in Definition \ref{def:strong} for $\mu = n - k$, using $F = \mathbb{F}_q$.

\subsection{Universal Secure Network Code}

The definition of strong security depends on the GCVs of the set of tapped links $\mathcal{I}$, implying that it is dependent on the underlying network code. Silva et al.\ proposed a coding scheme that attains strong security that is independent of the network code, as defined below.

\begin{definition}[universal strong security~\cite{SILVA}]
\begin{eqnarray}
&&H(S|Y) = 0, \label{eqn:USS1} \\
&&I(S;W=BX) = 0, \forall B \in \mathbb{F}_{q}^{\mu \times n}.\label{eqn:USS2}
\end{eqnarray}
\end{definition}

\noindent The universal code is based on MRD codes\cite{GABIDULIN} and coset coding scheme\cite{WYNER}. MRD codes are a class of linear code over $\mathbb{F}_{q^m}$ which is optimal in the rank-distance sense. Coset coding scheme is a type of randomized coding described as follows. Let $H$ be the parity check matrix of a $[n,n - k]$ linear code $\mathcal{C}$ over $F$. To code $S=(S_1, \cdots , S_k) \in F^k$ into $X=(X_1, \cdots, X_n) \in F^n$, regard $S$ as a syndrome of $\mathcal{C}$, and choose $X$ uniformly random from the corresponding coset. Using these tools, the  communication procedure of the universal network code is briefly described as follows:

The procedures of secret communication using the universal secure network code is briefly described as follows:
\begin{enumerate}
\item Choose an integer $m \geq n$.
\item Construct an $[n,\mu = n - k]$ MRD code over $\mathbb{F}_{q^m}$.
\item Encode $S \in \mathbb{F}_{q^m}^k\rightarrow X \in \mathbb{F}_{q^m}^n$ by coset coding scheme based on the MRD code.
\item Split $X$ and send them over $m$ time slots using a feasible network code, i.e. transmit  $(X_1^{(t)},X_2^{(t)},\cdots,X_n^{(t)})^T \in \mathbb{F}_q^n$ at time $1 \leq t \leq m$.
\end{enumerate}

\section{Stronger Eavesdropping Model}

In this section, we propose a stronger eavesdropping model than the one presented in Section \ref{sec:SNC}.

\subsection{Model Definition}

In the conventional non-universal secure network coding scheme, $F = \mathbb{F}_q$ is used, but note that in the universal scheme, due to the use of MRD code over $\mathbb{F}_{q^m}$, $F = \mathbb{F}_{q^m}$ is used. Since the network can only transmit up to $n$ elements of $\mathbb{F}_q$ per unit time, the universal code requires that a secret message $S$ be transmitted over multiple time slots, while the conventional codes require only one. The definition of the wiretap network model implies that the universal code assumes the selection of tapped links to be fixed during the transmission. Hence, we replace the eavesdropper model presented in Section \ref{sec:SNC} with the following stronger model. 

\vspace{1mm}
\noindent \textbf{Stronger Eavesdropper:} At each time slot of the transmission over $m$ time slots, the wiretapper can re-select the set of $\mu$ tapping links. Let $e_{i,t} \in E$ denote the $i$-th link tapped at time $t$. The wiretapped links are then,  $e_{1,1}, e_{2,1}, \cdots, e_{\mu,1}, \cdots\cdots, e_{1,m}, e_{2,m}, \cdots, e_{\mu,m}$. For $x = (x_1, x_2, \cdots, x_n)^T \in \mathbb{F}_{q^m}^n$, define $\bar{x} \in \mathbb{F}_{q}^{mn}$ as
\begin{equation*}
\bar{x} \triangleq (x_1^{(1)},x_2^{(1)},\cdots, x_n^{(1)}, \cdots \cdots, x_1^{(m)},x_2^{(m)},\cdots, x_n^{(m)})^T.
\end{equation*}
Note that there is a one-to-one correspondence between $x$ and $\bar{x}$. For simplification, let $\vec{b}_{i,t} \triangleq \vec{b}_{e_{i,t}}$. The GCVs of the $\mu$ links tapped at time $t$ are $\vec{b}_{1,t},\cdots, \vec{b}_{\mu,t} \in \mathbb{F}_q^{n}$. Also, define $\tilde{B} \in \mathbb{F}_q^{m\mu \times mn}$ and $B_t \in \mathbb{F}_q^{\mu \times n}$ as follows:
\begin{equation*}
\tilde{B} \triangleq 
\left[
\begin{array}{cccc}
B_1&&&\\
&B_2&&\\
&&\ddots&\\
&&&B_m
\end{array}
\right],\quad
B_t \triangleq 
\left[
\begin{array}{c}
\vec{b}_{1,t}\mbox{}^T\\
\vec{b}_{2,t}\mbox{}^T\\
\vdots\\
\vec{b}_{\mu,t}\mbox{}^T\\
\end{array}
\right].
\end{equation*}

\noindent Then, the information obtained by the wiretapper is represented by the random variable $\tilde{W}$ distributed over $\mathbb{F}_q^{m\mu}$, defined by
\begin{equation*}
\tilde{W} \triangleq \tilde{B}\bar{X}.
\end{equation*}

\noindent We now define the following security conditions that assure security against our eavesdropping model.
\begin{definition}[universal $m$-strong security]\label{def:m-univ}
\begin{eqnarray*}
&&H(S|Y) = 0,\\
&&I(S;\tilde{W}=\tilde{B}\bar{X}) = 0,\forall B_t \in \mathbb{F}_q^{\mu \times n}, t = 1, \cdots, m.
\end{eqnarray*}
\end{definition}

\noindent Note that, the conventional eavesdropping model defined in Section \ref{sec:SNC} corresponds to the special case of our model with $\vec{b}_{i,t_1} = \vec{b}_{i,t_2}, \forall t_1,t_2,i $. Also note that the security of the non-universal conventional secure network codes such as the one by Cai~\cite{CAI}, is not affected by such a strengthening of the eavesdropper because a secret message is transmitted over only one time slot. To be fair with the universal code, we also mention that even when $m$ secret messages are regarded as one message and are sent over $m$ time period, the conventional non-universal codes remain secure. To avoid confusion, we mention that universal $m$-strong security and $k$-strong security\cite{HARADA} are distinct notions.

\subsection{Code Example}

We present an example of Silva et al.'s universal secure network code and show that it is insecure against our eavesdropping model. The example code is constructed using the following parameters.

\begin{itemize}
\item $q = 2$, $k = 1$, $n = 2$, $m = 2, \mu = n - k = 1$.
\item $\mathbb{F}_{2^2}$ constructed with the root $\alpha$ of primitive polynomial $f(x) = x^2 + x + 1$. (Table \ref{table:field} shows the elements of this field in power, polynomial, and vector representation)
\begin{table}[t]
\begin{center}
\caption{The elements of $\mathbb{F}_{2^2}$}
\label{table:field}
\vspace{2mm}
\begin{tabular}[ht]{|c|c|c|}
\hline
Power & Polynomial & Vector\\
\hline
\hline
Zero & $0$ & $(0,0)$\\
\hline
$\alpha^0$ & $1$ & $(0,1)$\\
\hline
$\alpha^1$ & $\alpha^1$ & $(1,0)$\\
\hline
$\alpha^2$ & $\alpha^1 + 1$ & $(1,1)$\\
\hline
\end{tabular}
\end{center}
\end{table}
\item A parity check matrix $H = [1,\alpha]$ of a $[2,1]$MRD code over $\mathbb{F}_{2^2}$. 
\end{itemize}

\noindent Note that between $X = (X_1, X_2)^T$ and $S$, we have the relation 
\begin{equation}
S = HX = X_1 + \alpha X_2.
\end{equation}
This code uses a network code over $\mathbb{F}_2$ at rate $2$, so it is sufficient to consider only the links $e_1, e_2, e_3$ with GCVs $\vec{b}_{e_1} = (0,1)^T$, $\vec{b}_{e_2} = (1,0)^T$, $\vec{b}_{e_3} = (1,1)^T$. This implies that the information flowing on an arbitrary link is one of $X \cdot \vec{b}_{e_1} = X_1$, $X \cdot \vec{b}_{e_2} = X_2$, or $X \cdot \vec{b}_{e_3} = X_1 + X_2$. Table \ref{table:example} shows the value, represented in power and vector form, on each link with all distinct GCVs for each $X$ sent. The value of $S$ is also shown. 

\begin{table}[t]
\begin{center}
\caption{The value flowing on each link and $S$, for each $X$}
\label{table:example}
\vspace{2mm}
\begin{tabular}[c]{|c|c|c|c|}
\hline
$X_1$ & $X_2$ & $X_1 + X_2$ &  $S$\\
\hline
\hline
$0 =(\underline{0},0)$ & $0 =(0,0)$ & $0 =(0,0)$ & $0$ \\
$0 =(\underline{0},0)$ & $\alpha^0 =(0,1)$ & $\alpha^0 =(0,\underline{1})$ & $\alpha^1$\\
$0 =(\underline{0},0)$ & $\alpha^1 =(1,0)$ & $\alpha^1 =(1,0)$ & $\alpha^2$\\
$0 =(\underline{0},0)$ & $\alpha^2 =(1,1)$ & $\alpha^2 =(1,\underline{1})$ & $\alpha^0$\\
\hline
$\alpha^0 =(\underline{0},1)$ & $0 =(0,0)$ & $\alpha^0 =(0,\underline{1})$ & $\alpha^0$\\
$\alpha^0 =(\underline{0},1)$ & $\alpha^0 =(0,1)$ & $0 =(0,0)$ & $\alpha^2$\\
$\alpha^0 =(\underline{0},1)$ & $\alpha^1 =(1,0)$ & $\alpha^2 =(1,\underline{1})$ & $\alpha^1$\\
$\alpha^0 =(\underline{0},1)$ & $\alpha^2 =(1,1)$ & $\alpha^1 =(1,0)$ & $0$\\
\hline
$\alpha^1 =(1,0)$ & $0 =(0,0)$ & $\alpha^1 =(1,0)$ & $\alpha^1$\\
$\alpha^1 =(1,0)$ & $\alpha^0 =(0,1)$ & $\alpha^2 =(1,\underline{1})$ & $0$\\
$\alpha^1 =(1,0)$ & $\alpha^1 =(1,0)$ & $0 =(0,0)$ & $\alpha^0$\\
$\alpha^1 =(1,0)$ & $\alpha^2 =(1,1)$ & $\alpha^0 =(0,\underline{1})$ & $\alpha^2$\\
\hline
$\alpha^2 =(1,1)$ & $0 =(0,0)$ & $\alpha^2 =(1,\underline{1})$  & $\alpha^2$\\
$\alpha^2 =(1,1)$ &$\alpha^0 =(0,1)$  & $\alpha^1 =(1,0)$ & $\alpha^0$\\
$\alpha^2 =(1,1)$ &$\alpha^1 =(1,0)$  & $\alpha^0 =(0,\underline{1})$ & $0$\\
$\alpha^2 =(1,1)$ &$\alpha^2 =(1,1)$  & $0 =(0,0)$ & $\alpha^1$\\
\hline
\end{tabular}
\end{center}
\end{table}

An eavesdropper capable of re-selecting the tapping links at each time is able to wiretap an element of $\{(P^{(1)}, Q^{(2)})\mid P,Q \in \{X_1,X_2,(X_1 + X_2)\}\}$. Recall that $P^{(i)}$ represents the $i$-th element of the vector representation of $P \in \mathbb{F}_{q^m}$. When the sequence $(X_1^{(1)}, (X_1 + X_2)^{(2)}) = (0,1)$ (underlined on the table) is wiretapped, the candidates for $S$ are narrowed down to $\alpha^0, \alpha^1$, implying
\begin{eqnarray*}
&&H(S|X_1^{(1)},(X_1 + X_2)^{(2)}) \neq H(S)\\
&\Rightarrow& I(S;X_1^{(1)},(X_1 + X_2)^{(2)}) \neq 0\\
&\Rightarrow& I(S;\tilde{W} = \tilde{B}\bar{X}) \neq 0, \text{for some } \tilde{B}, \text{rank} \tilde{B} = 2.
\end{eqnarray*}
Therefore, we can conclude that this code does not attain universal $m$-strong security.

\section{Security Analysis}

In this section, we analyze the security of the universal secure network code against our stronger eavesdropping model. The example presented in the previous section shows that the universal code is not universal $m$-strongly secure \emph{in general}. Construction of the universal code involves the choice of parameters $n,k,q,m$, a parity check matrix $H$, and a basis of $\mathbb{F}_{q^m}$. A natural question to ask at this point is, if it is possible to secure this code by restricting these parameters. We show that universal $m$-strong security cannot be attained no matter how they are chosen. We also analyze the cases with a restricted number of tapping links and tapping duration.

\subsection{Proof of Vulnerability for $\mu = n - k$}

As a preparation, we first derive the necessary and sufficient condition for the universal code to be universal $m$-strongly secure. Let
\begin{equation*}
N_{s,w}^{\tilde{B}} \triangleq |\{x \in \mathbb{F}_{q^m}^n\mid s = Hx, w = \tilde{B}\bar{x}\}|.
\end{equation*}
\begin{lemma}\label{lem:nk}
The necessary and sufficient condition for the universal coding scheme with parameters $n,k,q,m,H$ and a fixed basis of $\mathbb{F}_{q^m}$ to attain universal $m$-strong security for $\mu \geq 1$ is, for $\forall w \in \mathbb{F}_{q}^{m\mu}$,$\forall B_t \in \mathbb{F}_q^{\mu \times n}, \text{rank}B_t = \mu, 1 \leq t \leq m$ the following holds:
\begin{equation*}
N_{s,w}^{\tilde{B}} = N_{s',w}^{\tilde{B}}, \forall s, s' \in \mathbb{F}_{q^m}^k.
\end{equation*}
\end{lemma}
\begin{proof}By the definition of universal $m$-strong security, for $\forall w \in \mathbb{F}_q^{m\mu}$,
\begin{eqnarray}
&&I(S;\tilde{W}) = 0 \nonumber \\
&\Leftrightarrow& \text{Pr}(S = s| \tilde{W} = w) = \text{Pr}(S = s), \forall s \in \mathbb{F}_{q^m}^k \nonumber \\
&\Leftrightarrow&\frac{|\{x \in \mathbb{F}_{q^m}^n\mid s=Hx,w=\tilde{B}\bar{x}\}|}{|\{x \in \mathbb{F}_{q^m}^n\mid w=\tilde{B}\bar{x}\}|} = \frac{1}{q^{mk}}, \forall s \label{eqn:lemnka}\\
&\Leftrightarrow& N_{s,w}^{\tilde{B}} = \frac{|\{x \in \mathbb{F}_{q^m}^n\mid w=\tilde{B}\bar{x}\}|}{q^{mk}}, \forall s. \nonumber
\end{eqnarray}
\noindent Equation (\ref{eqn:lemnka}) holds because $X$ is distributed uniformly over $\mathbb{F}_{q^m}^n$ and $S$ is distributed uniformly over $\mathbb{F}_{q^m}^k$. Note that to attain universal $m$-strong security, it is sufficient to satisfy the security condition for all full-rank $B_t, 1 \leq t \leq m$.
\end{proof}
\noindent We prove the vulnerability for the special case $\mu = n - k$, which corresponds to the case considered in the work by Silva et al. 
\begin{lemma}\label{lem:1}
The necessary and sufficient condition for the universal coding scheme with parameters $n,k,q,m,H$ and a fixed basis of $\mathbb{F}_{q^m}$ to attain universal $m$-strong security for $\mu = n - k$ is, for $\forall w \in \mathbb{F}_{q}^{m\mu}$, $\forall B_t \in \mathbb{F}_q^{\mu \times n}$, $\text{rank}B_t = \mu, 1 \leq t \leq m$, $\mathcal{X}_{w} = \{x\in\mathbb{F}_{q^m}^{n}\mid w=\tilde{B}\bar{x}\}$, the following holds:
\begin{equation*}
x \neq x' \Rightarrow Hx \neq Hx', \forall x,x' \in \mathcal{X}_w.
\end{equation*}
\end{lemma}
\begin{proof}
By Lemma \ref{lem:nk}, $\forall w \in \mathbb{F}_q^{m\mu}$,
\begin{eqnarray}
&& N_{s,w}^{\tilde{B}} = \frac{|\{x \in \mathbb{F}_{q^m}^n\mid w=\tilde{B}\bar{x}\}|}{q^{mk}}, \forall s \nonumber\\
&\Leftrightarrow& |\{x \in \mathcal{X}_w\mid s=Hx\}| = 1,\forall s \label{eqn:lem1b}\\
&\Leftrightarrow& x \neq x' \Rightarrow Hx \neq Hx', \forall x, x' \in \mathcal{X}_w \nonumber.
\end{eqnarray}
\noindent Equation (\ref{eqn:lem1b}) holds since 
\begin{eqnarray*}
|\{x \in \mathbb{F}_{q^m}^n\mid w = \tilde{B}\bar{x}\}| &=& q^{\dim \ker \tilde{B}}\\
&=& q^{(mn - \text{rank}\tilde{B})}\\
&=& q^{(mn - m(n-k))} = q^{mk}.
\end{eqnarray*}
\end{proof}

\noindent Note that if Lemma \ref{lem:1} holds for set $\mathcal{X}_w$ then the lemma holds for any of its subsets. Let $w_{i,t} \in \mathbb{F}_q$ be the information tapped at time $t$ on the $i$-th link. Then, by representing $w$ as
\begin{equation*}
w = (w_{1,1}, w_{2,1}, \cdots, w_{\mu,1}, \cdots \cdots, w_{1,m}, w_{2,m}, \cdots w_{\mu,m})^T, 
\end{equation*}
\noindent Lemma \ref{lem:1} yields the following corollary.

\begin{corollary}\label{corollary}
The necessary and sufficient condition for the universal coding scheme with parameters $n,k,q,m,H$ and a fixed basis of $\mathbb{F}_{q^m}$ to satisfy universal $m$-strong security for $\mu = n - k$ is that
\begin{equation*}
x \neq x' \Rightarrow Hx \neq Hx', \forall x, x' \in \mathcal{X}
\end{equation*}
holds for an arbitrary set $\mathcal{X} \subseteq \mathbb{F}_{q^m}^n$ such that $\forall x \in \mathcal{X}$ satisfies
\begin{equation*}
\left\{
\begin{array}{l}
(\vec{b}_{1,1}\cdot x)^{(1)} = w_{1,1}, \cdots, (\vec{b}_{1,m}\cdot x)^{(m)} = w_{1,m},\\
(\vec{b}_{2,1}\cdot x)^{(1)} = w_{2,1}, \cdots, (\vec{b}_{2,m}\cdot x)^{(m)} = w_{2,m},\\
\qquad \vdots\\
(\vec{b}_{\mu,1}\cdot x)^{(1)} = w_{\mu,1}, \cdots ,(\vec{b}_{\mu,m}\cdot x)^{(m)} = w_{\mu, m},
\end{array}
\right.
\end{equation*}\label{eqn:corollary}
for $\forall w \in \mathbb{F}_q^{m\mu}$.
\end{corollary}
\begin{proof}
By denoting the $l$-th element of $\vec{b}_{i,t}$ as $b_{i,t}^{[l]}\in \mathbb{F}_q$,
\begin{eqnarray}
(\vec{b}_{i,t}\cdot x)^{(t)} &=& (b_{i,t}^{[1]}x_1 + b_{i,t}^{[2]}x_2 + \cdots + b_{i,t}^{[n]}x_n)^{(t)} \nonumber \\
&=& (b_{i,t}^{[1]}x_1^{(t)}+b_{i,t}^{[2]}x_2^{(t)}+ \cdots + b_{i,t}^{[n]}x_n^{(t)})\label{eqn:t-th}\\ 
&=& \vec{b}_{i,t} \cdot (x_1^{(t)},x_2^{(t)}, \cdots ,x_n^{(t)})^T \nonumber 
\end{eqnarray}
\noindent holds. Note that since $\mathbb{F}_{q^m}$ is a linear space on $\mathbb{F}_q$, 
\begin{eqnarray*}
b_{i,t}^{[l]} x_l &=& b_{i,t}^{[l]} (x_l^{(1)},x_l^{(2)}, \cdots , x_l^{(m)})^T\\
&=& (b_{i,t}^{[l]}x_l^{(1)}, b_{i,t}^{[l]}x_l^{(2)}, \cdots , b_{i,t}^{[l]}x_l^{(m)})^T
\end{eqnarray*}
holds for every $1 \leq l \leq n$, and adding the $t$-th element of each of $b_{i,t}^{[1]} x_1, \cdots , b_{i,t}^{[n]} x_n$ yields Eq.~(\ref{eqn:t-th}). Therefore, we have the relation,
\begin{equation}
w_{i,t} =  \vec{b}_{i,t} \cdot (x_1^{(t)},x_2^{(t)},\cdots,x_n^{(t)})^T  = (\vec{b}_{i,t} \cdot x)^{(t)}.
\label{eqn:inner}
\end{equation}

\noindent The corollary holds immediately from Eq.~(\ref{eqn:inner}) and Lemma \ref{lem:1}.
\end{proof}

\noindent Using this corollary, we prove the following theorem.

\begin{theorem}\label{thm:1}
For any choice of parameters $n,k,q,m,H$ and the basis for $\mathbb{F}_{q^m}$, the universal secure network coding scheme cannot satisfy the universal $m$-strong security condition for $\mu = n - k$.
\end{theorem}
\begin{proof}
Assume the existence of the universal code which satisfies the universal $m$-strong security condition. Let $\mathcal{X}$ denote the set of all $x \in \mathbb{F}_{q^m}^n$ satisfying the relation, 
\begin{equation}
\left\{
\begin{array}{l}
x_1^{(1)} = x_1^{(2)} = \cdots = x_1^{(m)} = 0,\\
x_2^{(1)} = x_2^{(2)} = \cdots = x_2^{(m)} = 0,\\
\qquad \vdots\\
x_\mu^{(1)} = x_\mu^{(2)} = \cdots = x_\mu^{(m)} = 0.
\end{array}
\right.
\label{eqn:X}
\end{equation}
\noindent Note that, $|\mathcal{X}| = q^{mn} / q^{m\mu} = q^{mk}$. Let $\alpha$ be an element of $\mathbb{F}_q$, and choose $\hat{x} \in \mathbb{F}_{q^m}^n$ that satisfies the following:
\begin{equation*}
\left\{
\begin{array}{l}
\hat{x}_1^{(1)} = \hat{x}_1^{(2)} = \cdots = \hat{x}_1^{(m)} = 0,\\
\hat{x}_2^{(1)} = \hat{x}_2^{(2)} = \cdots = \hat{x}_2^{(m)} = 0,\\
\qquad \vdots\\
\hat{x}_\mu^{(1)} = \hat{x}_\mu^{(2)} = \cdots = \hat{x}_{\mu}^{(m - 1)} = 0, \hat{x}_{\mu}^{(m)} = 1, \hat{x}_{\mu + 1}^{(m)} = \alpha.
\end{array}
\right.
\end{equation*}
\noindent Such $\hat{x}$ always exists, and satisfies $\hat{x} \not\in \mathcal{X}$. For $\psi \in \mathbb{F}_q$, let $\mathcal{X}_\psi \subseteq \mathcal{X}$ be the set of all $x \in \mathcal{X}$ satisfying $x_{\mu + 1}^{(m)} = \psi$. In other words, $\forall x \in \mathcal{X}_\psi$ satisfies Eq.~(\ref{eqn:Xa}).
\begin{equation}
\left\{
\begin{array}{l}
x_1^{(1)} = x_1^{(2)} = \cdots = x_1^{(m)} = 0,\\
x_2^{(1)} = x_2^{(2)} = \cdots = x_2^{(m)} = 0,\\
\qquad \vdots\\
x_\mu^{(1)} = x_\mu^{(2)} = \cdots = x_{\mu}^{(m)} = 0,  x_{\mu + 1}^{(m)} = \psi.
\end{array}
\right.
\label{eqn:Xa}
\end{equation}
\noindent Note that,
\begin{equation}
\bigcup_{\psi \in \mathbb{F}_q} \mathcal{X}_\psi = \mathcal{X} \label{eqn:prop1} 
\end{equation}
\noindent holds. We see that,
\begin{equation}
\left\{
\begin{array}{l}
x_1^{(1)} = x_1^{(2)} = \cdots = x_1^{(m)} = 0,\\
x_2^{(1)} = x_2^{(2)} = \cdots = x_2^{(m)} = 0,\\
\qquad \vdots\\
x_\mu^{(1)} = x_\mu^{(2)} = \cdots = x_{\mu}^{(m-1)} = 0, (\gamma x_\mu +  x_{\mu + 1})^{(m)} = \psi,
\end{array}
\right.
\label{eqn:Xalphai}
\end{equation}
\noindent holds for $\forall x \in \{\hat{x}\} \cup \mathcal{X}_\psi$ by the definition of $\hat{x}$ and $\mathcal{X}_\psi$, where $\gamma = \psi - \alpha$. Note that, Corollary \ref{corollary} can be applied to the set $\{\hat{x}\} \cup \mathcal{X}_\psi$ because relation (\ref{eqn:Xalphai}) can be represented as 
\begin{equation*}
\left\{
{\footnotesize 
\begin{array}{l}
((1,0,\cdot\cdots\cdots\cdots\cdots,0)\cdot x)^{(1)} = 0,\cdots,((1,0,\cdots\cdots\cdots\cdots\cdot\cdot\cdot\,\,,0)\cdot x)^{(m)} = 0,\\
((0,1,0,\cdot\cdots\cdots\cdots\cdot\,\,,0)\cdot x)^{(1)} = 0,\cdots, ((0,1,0,\cdots\cdots\cdots\cdot\cdot\,\cdot\,\,,0)\cdot x)^{(m)} = 0,\\
\hspace{37mm}\vdots\\
((0,\cdot\cdot,0,\underbrace{1}_{\mu\text{-th}},0,\cdot\cdot,0)\cdot x)^{(1)} = 0,\cdots, ((0,\cdot\cdot,0,\underbrace{\gamma}_{\mu\text{-th}},1,0,\cdot\cdot,0)\cdot x)^{(m)} = \psi,
\end{array}
}
\right.
\end{equation*}
\noindent and because it can be confirmed that the corresponding matrices $B_t$ in Corollary \ref{corollary} satisfy $\text{rank}B_t = \mu$, $1 \leq t \leq m$. Applying Corollary \ref{corollary} to the set $\{\hat{x}\} \cup \mathcal{X}_\psi$ we have 
\begin{equation*}
x \neq x' \Rightarrow Hx \neq Hx', \forall x,x' \in \{\hat{x}\} \cup \mathcal{X}_\psi,\forall \psi \in \mathbb{F}_q.
\end{equation*}
\noindent This result, combined with Eq.~(\ref{eqn:prop1}), yields
\begin{equation}
H\hat{x} \neq Hx, \forall x \in \mathcal{X}. \label{eqn:theo3a}
\end{equation}
\noindent Corollary \ref{corollary} can be applied to the set $\mathcal{X}$ as well by relation (\ref{eqn:X}), hence
\begin{equation}
x \neq x' \Rightarrow Hx \neq Hx', \forall x, x' \in \mathcal{X}\label{eqn:theo3b}
\end{equation}
holds. Therefore, Eqs.~(\ref{eqn:theo3a}) and (\ref{eqn:theo3b}) yield
\begin{equation}
x \neq x' \Rightarrow Hx \neq Hx', \forall x,x' \in \{\hat{x}\} \cup \mathcal{X}.
\label{eqn:last}
\end{equation}
\noindent However, by $\hat{x} \not \in \mathcal{X}$ and $|\mathcal{X}| = q^{mk}$, for Eq.~(\ref{eqn:last}) to hold, it is necessary that
\begin{equation*}
|\{Hx \mid x \in \mathbb{F}_{q^m}^n\}|\geq |\{\hat{x}\} \cup \mathcal{X}| = q^{mk} + 1
\end{equation*}
\noindent holds, which contradicts with
\begin{equation*}
|\{Hx\mid x \in \mathbb{F}_{q^m}^n\}| = {(q^m)}^{\text{rank}H} = q^{mk},
\end{equation*}
\noindent Therefore, a code constructed by the universal coding scheme that attains universal $m$-strong security does not exist.
\end{proof}

\subsection{Proof of Vulnerability for $1 \leq \mu \leq n - k$}
We now consider the more general case, $1 \leq \mu \leq n - k$, and prove that the code still cannot be secure for any $\mu$. First, we define the following to simplify the notations:
\begin{itemize}
\item $\mathcal{H}_s \triangleq \{x \in \mathbb{F}_{q^m}^n \mid s = Hx \}$,
\item $\mathcal{X}_w^{\tilde{B}} \triangleq \{x \in \mathbb{F}_{q^m}^n\mid w = \tilde{B}\bar{x}\}$.
\end{itemize}
\noindent Now we prove the following lemma.
\begin{lemma}\label{lem:3}
If the universal coding scheme is universal $m$-strongly secure for $\mu = 1$, then it is universal $m$-strongly secure for $\mu = n - k$.
\end{lemma}
\begin{proof}
When universal $m$-strong security for $\mu = 1$ is attained for some code, by Lemma \ref{lem:nk}, for $\forall w \in \mathbb{F}_{q}^{m}$, $\forall B_t \in \mathbb{F}_q^{1 \times n}$, $\text{rank}B_t = 1$, $1 \leq t \leq m$, the following must hold for this code:
\begin{equation*}
N_{s,w}^{\tilde{B}} = N_{s',w}^{\tilde{B}}, \forall s, s' \in \mathbb{F}_{q^m}^k.
\end{equation*} \label{eqn:asm}
\noindent We will show that, then, such a code attains universal $m$-strongly security for $\mu = n - k$, which implies that by Lemma \ref{lem:nk}, for
$\forall w^* \in \mathbb{F}_{q}^{m(n-k)}$,$\forall B^*_t \in \mathbb{F}_q^{(n-k) \times n}$, $\text{rank}B^*_t = n-k$, $1 \leq t \leq m$,
\begin{equation*}
B^* = 
\left[
\begin{array}{cccc}
B_1^*&&&\\
&B_2^*&&\\
&&\ddots&\\
&&&B_m^*
\end{array}
\right],
\end{equation*}
\noindent the following holds:
\begin{equation*}
N_{s,w^*}^{B^*} = N_{s',w^*}^{B^*}, \forall s, s' \in \mathbb{F}_{q^m}^k.
\end{equation*}
\noindent Let $B^*_{t,i} \in \mathbb{F}_q^{1 \times n}$ denote the $i$-th row of $B^*_t$, and let $B'$ denote the matrix defined as below, using the matrices $B^{[i]} \in \mathbb{F}_q^{m \times mn}$, $1 \leq i \leq n - k$:
\begin{equation*}
B' = 
\left[
\begin{array}{c}
B^{[1]}\\
B^{[2]}\\
\vdots \\
B^{[n-k]}\\
\end{array}
\right],
B^{[i]} = 
\left[
\begin{array}{cccc}
B_{1,i}^*&&&\\
&B_{2,i}^*&&\\
&&\ddots&\\
&&&B_{m,i}^*
\end{array}
\right].
\end{equation*}
\noindent Note that $B'$ is obtained by permuting the rows of $B^*$. Let $w'$ be the column vector obtained by permuting the rows of $w^*$ in the same order as $B'$, and let $w^{[i]} \in \mathbb{F}_q^{m}, 1 \leq i \leq n - k$, denote each $m$ rows of $w'$ as shown below:
\begin{equation}
w' = 
\left[
\begin{array}{c}
w^{[1]}\\
w^{[2]}\\
\vdots \\
w^{[n-k]}\\
\end{array}
\right].
\label{eqn:form}
\end{equation}

\noindent Using the notations above, we have
\begin{eqnarray*}
N_{s,w^*}^{B^*} &=& |\{x \in \mathbb{F}_{q^m}^n\mid s = Hx, w^* = B^*\bar{x}\}|\\ 
&=& |\{x \in \mathbb{F}_{q^m}^n\mid s = Hx, w' = B'\bar{x}\}|\\ 
&=& \left| \left\{x \in \mathcal \mathbb{F}_{q^m} \mid s = Hx,
\left[
\begin{array}{c}
w^{[1]}\\
w^{[2]}\\
\vdots \\
w^{[n-k]}\\
\end{array}
\right]
=
\left[
\begin{array}{c}
B^{[1]}\\
B^{[2]}\\
\vdots \\
B^{[n-k]}\\
\end{array}
\right]
\bar{x}\right\}\right| \\
&=& \left| \bigcap_{i = 1}^{n-k} \left\{ x \in \mathcal{X}_{w^{[i]}}^{B^{[i]}} \mid s = Hx \right\} \right|\\
&=& \left| \left(\bigcap_{i = 1}^{n-k} \mathcal{X}_{w^{[i]}}^{B^{[i]}}\right) \bigcap \mathcal{H}_s \right|.
\end{eqnarray*}

\noindent Noting that
\begin{equation}
\left|\cap_{i = 1}^{n-k} \mathcal{X}_{w^{[i]}}^{B^{[i]}}\right|  = \left|\left\{ w^* = B^*\bar{x}  \right\}\right| = q^{mk}, \label{eqn:setsize}
\end{equation}
\noindent we prove $N_{s_1,w^*}^{B^*} = N_{s_2,w^*}^{B^*}, \forall s_1, s_2 \in \mathbb{F}_{q^m}^k$ for each of the following three cases of the set $\left(\bigcap_{i = 1}^{n-k} \mathcal{X}_{w^{[i]}}^{B^{[i]}}\right) \bigcap \mathcal{H}_s $.   
\\
\noindent \emph{Case 1} $\left(\bigcap_{i = 1}^{n-k} \mathcal{X}_{w^{[i]}}^{B^{[i]}}\right) \bigcap \mathcal{H}_s = \phi, \forall s$: By $\bigcup_s \mathcal{H}_s = \mathbb{F}_{q^m}^n$, we have
\begin{equation*}
\left(\bigcap_{i = 1}^{n-k} \mathcal{X}_{w^{[i]}}^{B^{[i]}}\right) \bigcap \mathcal{H}_s = \phi, \forall s \Leftrightarrow \bigcap_{i = 1}^{n-k} \mathcal{X}_{w^{[i]}}^{B^{[i]}} = \phi,
\end{equation*}
\noindent which contradicts with Eq.~(\ref{eqn:setsize}). Thus, this case does not exist.
\\
\noindent \emph{Case 2} $\left(\bigcap_{i = 1}^{n-k} \mathcal{X}_{w^{[i]}}^{B^{[i]}}\right) \bigcap \mathcal{H}_s \neq \phi, \forall s$: Clearly,
\begin{eqnarray}
\left| \left( \bigcap_i \mathcal{X}_{w^{[i]}}^{B^{[i]}}\right) \cap \mathcal{H}_{s}\right| \geq 1, \forall s \label{eqn:b}
\end{eqnarray}
\noindent holds. By Eq.~(\ref{eqn:setsize}), we have
\begin{eqnarray*}
&& \left| \left( \bigcap_i \mathcal{X}_{w^{[i]}}^{B^{[i]}}\right) \cap \bigcup_{s \in \mathbb{F}_{q^m}^k}\mathcal{H}_{s}\right| \leq q^{mk}\\
&\Leftrightarrow& \left| \bigcup_{s \in \mathbb{F}_{q^m}^k} \left( \left( \bigcap_i \mathcal{X}_{w^{[i]}}^{B^{[i]}}\right) \cap \mathcal{H}_{s} \right) \right| \leq q^{mk}\\
&\Leftrightarrow& \left| \left( \bigcap_i \mathcal{X}_{w^{[i]}}^{B^{[i]}}\right) \cap \mathcal{H}_{s}\right| =  1, \forall s.
\end{eqnarray*}
\noindent The last line yields from Eq.~(\ref{eqn:b}), $|\mathbb{F}_{q^m}^k| = q^{mk}$, and
\begin{equation*}
\mathcal{H}_{s_1} \cap \mathcal{H}_{s_2} = \phi, \forall s_1, s_2 \in \mathbb{F}_{q^m}^k, s_1\neq s_2.
\end{equation*}
\noindent Thus, the lemma holds for this case.\\
\noindent \emph{Case 3} Otherwise: There exist $s_1, s_2 \in \mathbb{F}_{q^m}^k$ and $1 \leq l \leq n-k$ that satisfy the following:
\begin{eqnarray*}
\left\{
\begin{array}{c}
\mathcal{X}_{w^{[l]}}^{B^{[l]}} \bigcap \mathcal{H}_{s_1} = \phi\\
\mathcal{X}_{w^{[l]}}^{B^{[l]}} \bigcap \mathcal{H}_{s_2} \neq \phi
\end{array}
\right.
\Leftrightarrow
\left\{
\begin{array}{c}
N_{s_1,w^{[l]}}^{B^{[l]}} = 0\\
N_{s_2,w^{[l]}}^{B^{[l]}} \neq 0
\end{array}.
\right.
\end{eqnarray*}
\noindent However, $N_{s_1,w^{[l]}}^{B^{[l]}} \neq N_{s_2,w^{[l]}}^{B^{[l]}}$ contradicts with the assumption that universal $m$-strong security for $\mu = 1$ is satisfied. Thus, Case 3 does not exist.\\
\noindent We considered all three cases, which cover all possible cases and are disjoint, and conclude that the lemma holds since it holds for Case 2 which is the only existing case.
\end{proof}
Assume the existence of a secure network code that satisfies the universal $m$-strong security condition for some $1 \leq \mu \leq n-k$. Then this code must satisfy the security condition for $\mu = 1$ because in this case, the amount of information that can be wiretapped is obviously no more than the case for $1 \leq \mu \leq n - k$. Then by Lemma \ref{lem:3}, this code satisfies the security condition for $\mu = n - k$, which contradicts with Theorem \ref{thm:1} stating that universal $m$-strong security for $\mu = n - k$ cannot be attained. Hence, we have the following result.

\begin{theorem}
For any choice of parameters $n,k,q,m,H$ and the basis for $\mathbb{F}_{q^m}$, the universal secure network coding scheme cannot attain universal $m$-strong security for $1 \leq \mu \leq n - k$.
\end{theorem}

\subsection{Restricted Tapping Time}

Now we consider the case when the tapping duration is generalized to $1 \leq m' \leq m$ in addition to the generalized $\mu$ considered in the previous part. Since the tapping duration is restricted, we do not restrict $\mu$ to $1 \leq \mu \leq n - k$, and assume $1 \leq \mu$ instead. This imposes an additional condition, $B_t = O,\forall t \in M$ for any choice of $M \subseteq \{1,2, \cdots, m\}, |M| = m - m'$, on Definition \ref{def:m-univ}. Note that the set $M$ represents the set of time slot indices at which wiretapping does not occur. We are interested in, with which pairs of $\mu$ and $m'$ the universal code becomes secure. From the discussions up to this point and the result of Silva et al., the following is clear: 
\begin{itemize}
\setlength{\parskip}{0cm} 
\setlength{\itemsep}{0cm}
\item $\mu = 1$ and $1 \leq m' \leq n - k$: secure
\item $1 \leq \mu \leq n$ and $m' = m$: insecure
\item $1 \leq \mu \leq  n - k$ and $m' = 1$: secure
\end{itemize}

\noindent Additionally, by the necessary and sufficient condition of universal $m$-strong security in Lemma \ref{lem:1} we have, 
\begin{eqnarray}
&&N_{s,w}^{\tilde{B}} = \frac{|\{x \in \mathbb{F}_{q^m}^n\mid w=\tilde{B}\bar{x}\}|}{q^{mk}}, \forall s \nonumber \\
&\Leftrightarrow&N_{s,w}^{\tilde{B}} = \frac{q^{mn - m'\mu}}{q^{mk}}, \forall s \label{eqn:nec}
\end{eqnarray}
\noindent Since the RHS takes a positive value, and by the definition of $N_{s,w}^{\tilde{B}}$ the LHS must be a non-negative integer, a necessary condition for satisfying Eq.~(\ref{eqn:nec}), or the necessary condition for the code to be universal $m$-strongly secure, is as follows:
\begin{eqnarray}
&&\frac{q^{mn - m'\mu}}{q^{mk}} \geq 1 \nonumber \\
&\Leftrightarrow& q^{mn - m'\mu} \geq q^{mk} \nonumber \\
&\Leftrightarrow& m' \leq \frac{m(n-k)}{\mu}\label{eqn:bound}.
\end{eqnarray}
\noindent 

\noindent For any fixed $m'$, $m$, $n$, and $k$, Eq.~(\ref{eqn:bound}) is unsatisfied for $\mu \geq m(n-k) + 1$ because of the restriction $ m' \geq 1$. Note that being able to wiretap $n$ arbitrary links allows the wiretapper to obtain the maximum amount of information that can possibly be wiretapped over the $m'$ time slots, by continuously tapping the $n$ links with GCVs that form a basis of $\mathbb{F}_q^n$. Thus, the amount of information obtained by the wiretapper with $\mu \geq n$ is at least as much as what is obtained by the wiretapper with $\mu \geq m(n-k) + 1$, which implies that the code is insecure with $\mu \geq n$. Thus, we have another necessary condition,
\begin{eqnarray}
\mu \leq n - 1 \label{eqn:bound2}.
\end{eqnarray}
\noindent Combining Eqs.~(\ref{eqn:bound}) and (\ref{eqn:bound2}) yields the necessary condition, 
\begin{eqnarray*}
\mu \leq \text{min}\left\{ \frac{m}{m'}(n-k), n - 1\right\}.
\end{eqnarray*}

\section{Conclusion}

We proposed an eavesdropping model where the adversary is able to re-select the tapping wires at each time slot during the communication. We proved the impossibility of securing against this model using the universal secure network code proposed by Silva et al.\ for all choices of code parameters, even with a restricted number of tapped links. Moreover, we considered the case with shorter tapping duration, and derived a necessary condition for this code to be secure. The future tasks include improving this condition to a necessary and sufficient one. 

\section*{Acknowledgment}
The authors would like to thank the anonymous reviewers of ISIT and IEICE for their valuable comments and suggestions which improved our paper. They also thank Jun Kurihara for his comments during the seminars. This research was partly supported by the Japan Society for the Promotion of Science under Grants-in-Aid No.~20760233.

\end{document}